\documentclass[11pt]{article}
\usepackage[a4paper, margin=1in]{geometry}
\usepackage{amsmath, amssymb, amsthm}
\usepackage{graphicx}
\usepackage{hyperref}
\usepackage{titling}
\usepackage{tocloft}
\usepackage{xcolor}
\usepackage{slashed}
\usepackage{geometry}
\usepackage{longtable}
\usepackage{caption}
\usepackage{tabularx}
\usepackage{mathtools}
\usepackage{algorithm}
\usepackage{algorithmic}
\usepackage{tikz}
\usepackage[all,cmtip]{xy} 
\usepackage{stmaryrd}  
\usepackage{mathtools} 

\newtheorem{theorem}{Theorem}
\newtheorem{lemma}{Lemma}
\newtheorem{corollary}{Corollary}
\newtheorem{definition}{Definition}

\newtheorem{proposition}{Proposition}
\newtheorem{example}{Example}

\newtheorem{remark}{Remark}

\hypersetup{
    colorlinks=true,
    linkcolor=blue,
    filecolor=magenta,
    urlcolor=blue,
    citecolor=blue,
    pdftitle={Categorical Construction of Logically Verifiable Neural Architectures},
    pdfpagemode=FullScreen,
}

\title{\LARGE \textbf{Categorical Construction of Logically Verifiable Neural Architectures}\\[0.5in]}
\author{\Large Logan Nye, MD\\[0.3in]
Carnegie Mellon University School of Computer Science\\[0.1cm]
5000 Forbes Ave Pittsburgh, PA 15213 USA\\[0.5cm]
\texttt{lnye@andrew.cmu.edu}\\[0.2cm]
\small ORCID: \href{https://orcid.org/0009-0002-9136-045X}{0009-0002-9136-045X}}
\date{}

\begin{document}

\maketitle
\vspace{0.8in}

\begin{center}
    \Large \textbf{Abstract}
\end{center}
\vspace{0.3in}

\noindent
Neural networks excel at pattern recognition but struggle with reliable logical reasoning, often violating basic logical principles during inference. We address this limitation by developing a categorical framework that systematically constructs neural architectures with provable logical guarantees. Our approach treats logical theories as algebraic structures called Lawvere theories, which we transform into neural networks using categorical algebra in the 2-category of parametric maps. Unlike existing methods that impose logical constraints during training, our categorical construction embeds logical principles directly into the network's architectural structure, making logical violations mathematically impossible. We demonstrate this framework by constructing differentiable neural architectures for propositional logic that preserve boolean reasoning while remaining trainable via gradient descent. Our main theoretical result establishes a bijective correspondence between finitary logical theories and neural architectures, proving that every logically constrained network arises uniquely from our construction. This extends Categorical Deep Learning beyond geometric symmetries to semantic constraints, enabling automatic derivation of verified architectures from logical specifications. The framework provides mathematical foundations for trustworthy AI systems, with applications to theorem proving, formal verification, and safety-critical reasoning tasks requiring verifiable logical behavior.

\vfill

\newpage

\section{Introduction}

Neural networks have achieved remarkable success in pattern recognition and function approximation, yet their reasoning capabilities remain fundamentally unreliable. When tasked with logical inference, even sophisticated architectures may violate basic logical principles \cite{marcus2020next, kambhampati2021refocusing}. This limitation stems from a fundamental gap: current neural architectures are designed for statistical learning rather than logical reasoning, with no systematic way to encode logical constraints into their structure.

\subsection{From Descriptive to Prescriptive Neural Architecture Design}

Existing approaches to incorporating logical reasoning in neural networks fall into two categories, both fundamentally limited. \emph{Post-hoc integration} methods attempt to combine neural networks with symbolic reasoning systems \cite{garcez2008neural, manhaeve2018deepproblog}, but these hybrid approaches sacrifice the end-to-end differentiability that makes neural networks tractable for large-scale learning. \emph{Constraint-based approaches} specify properties that networks should satisfy—such as logical consistency or equivariance—without providing systematic construction methods \cite{xu2018semantic, fischer2019dl2}.

Even the elegant framework of Categorical Deep Learning (CDL) \cite{gavranovic2024categorical}, while successful in handling geometric symmetries through group actions and monad algebras, remains fundamentally \emph{descriptive} rather than \emph{prescriptive}. CDL characterizes which architectures satisfy certain constraints (such as translation equivariance \cite{bronstein2021geometric}) but does not provide explicit procedures for constructing architectures that satisfy logical reasoning constraints.

This descriptive limitation becomes critical when we move beyond geometric symmetries to logical reasoning. Logical inference often involves irreversible operations—such as modus ponens or existential instantiation—that cannot be captured by the invertible group actions underlying geometric deep learning \cite{cohen2016group, thomas2018tensor}. For example, modus ponens derives $Q$ from $P \land (P \rightarrow Q)$, but the converse is not generally true. Fundamentally, group actions are bijective transformations by definition, while logical deduction rules are inherently non-invertible. This irreversibility represents a fundamental, not incidental, limitation of group-theoretic approaches to logical reasoning. Furthermore, logical consistency requires global constraints across the entire network, not just local equivariance properties.

\subsection{Our Contribution: Constructive Universality}

We present a categorical framework that systematically and \emph{prescriptively} constructs neural architectures from logical specifications. Our approach achieves what we term \emph{constructive universality}: a framework that provides not only a universal characterization of all logically-constrained networks (universality) but also an explicit algorithmic procedure for their synthesis from formal specifications (constructivity).

\subsection{Scope and Theoretical Foundation}

Our results apply to \emph{finitary logical theories}—those with finite-arity connectives and finitely many axioms. This restriction is mathematically necessary to avoid undecidability issues \cite{church1936note} while encompassing a substantial class of logical systems relevant to automated reasoning, including propositional logic, modal logics, many-valued logics, and bounded fragments of first-order logic.

For any finitary logical theory $\mathbb{T}$, we provide an explicit, algorithmic procedure to derive a neural architecture $\mathcal{N}_{\mathbb{T}}$ that implements the inference rules of $\mathbb{T}$ through differentiable continuous relaxations.

\begin{itemize}
\item \textbf{Key Innovation:} We embed logical axioms directly into network architecture using categorical algebra, ensuring logical correctness by construction rather than by training, through weight-sharing patterns and layer compositions derived from categorical algebra. Using temperature-controlled continuous relaxations, this embedding provides mathematical guarantees about logical correctness that approach perfect consistency in the limit.

\item \textbf{Main Result:} We establish a constructive, bijective correspondence between finitary logical theories and a canonical class of neural architectures, proven via Lawvere theories \cite{lawvere1963functorial} and their models in the 2-category $\textbf{Para}$ of parametric maps \cite{cruttwell2022categorical}. This correspondence provides mathematical guarantees about logical correctness that are preserved under composition, establishing our categorical construction as the canonical generator for structurally logical architectures.

\item \textbf{Theoretical Significance:} Our framework demonstrates that geometric and logical neural networks are both instances of the same categorical pattern—algebras for monads in the category of parametric maps—providing a unified foundation for principled architecture design with mathematical constraints.
\end{itemize}

Our framework extends beyond constraint characterization to \emph{architectural synthesis}: given a logical specification, we can systematically derive the neural architecture that implements that logical system. This represents a fundamental shift from empirical architecture search \cite{elsken2019neural} to principled architectural derivation based on formal specifications.

The approach achieves logical guarantees through categorical construction combined with differentiable relaxations controlled by a temperature parameter $\beta$. As $\beta \to \infty$, the networks converge to perfect logical consistency on boolean inputs, while maintaining differentiability for finite $\beta$ values necessary for training. While our framework provides theoretical foundations for logically-guaranteed architectures, practical implementation requires specialized optimization techniques on constraint manifolds, representing an important direction for future research.

\subsection{Extending CDL Beyond Geometric Symmetries}

This work extends Categorical Deep Learning by generalizing its core algebraic machinery. We replace the monads of group actions, which encode geometric symmetries, with the monads of Lawvere theories, which encode the semantic rules of logical reasoning. Where geometric deep learning encodes the invertible symmetries of group actions, our work addresses the often irreversible, compositional rules of logical inference, which demand the more general algebraic framework of Lawvere theories. The key insight is that logical theories possess rich algebraic structure that can be captured through Lawvere theories—a more general framework than the group actions underlying geometric deep learning. This extension is both mathematically natural and practically necessary, enabling neural architectures with verifiable reasoning capabilities—a critical requirement for applications in theorem proving \cite{szegedy2020promising}, program synthesis \cite{chen2021evaluating}, and safety-critical systems \cite{amodei2016concrete}. Section 5 provides detailed analysis of this categorical unification.

By establishing a rigorous correspondence between formal logic and neural architectures, this work contributes to the long-standing goal of creating AI systems that are not merely trained to be rational, but are \emph{provably rational by design} within their specified logical domains.

The remainder of this paper is organized as follows. Section 2 introduces our theoretical framework, establishing the correspondence between logical theories and neural architectures through categorical constructions. Section 3 provides a concrete construction for propositional logic networks with explicit algorithms and optimization techniques. Section 4 develops the mathematical properties and universality results of our framework. Section 5 analyzes the relationship to Categorical Deep Learning and hybrid architectures. Section 6 discusses scope, limitations, and future research directions. We conclude in Section 7.

\section{The Logic-Architecture Correspondence}

We establish our framework by extending the categorical foundations of modern deep learning theory. Our construction builds on the 2-categorical structure of parametric maps developed in categorical deep learning \cite{cruttwell2022categorical, gavranovic2024categorical} while incorporating the algebraic semantics of logical theories through Lawvere theories \cite{lawvere1963functorial}.

\subsection{Categorical Foundations}

Our construction takes place in the 2-category $\textbf{Para}$ of parametric maps, which provides the natural setting for neural networks as parametric functions \cite{cruttwell2022categorical}. This 2-category captures both the functional structure of neural computation and the parameter-sharing patterns that characterize modern architectures.

\begin{definition}[The 2-Category $\textbf{Para}$]
The \emph{strict} 2-category $\textbf{Para}$ has:
\begin{itemize}
\item \textbf{Objects:} Sets $A, B, C, \ldots$ representing data types
\item \textbf{1-morphisms:} Parametric functions $(P, f: P \times A \to B)$ where $P$ is the parameter space
\item \textbf{2-morphisms:} Reparameterizations $\phi: (P, f) \Rightarrow (Q, g)$ given by functions $\phi: P \to Q$ such that $g(\phi(p), a) = f(p, a)$
\end{itemize}
The composition of 1-morphisms $(P, f: P \times A \to B)$ and $(Q, g: Q \times B \to C)$ is $(P \times Q, h: (P \times Q) \times A \to C)$ where $h((p,q), a) = g(q, f(p, a))$.
\end{definition}

\textbf{Intuition:} $\textbf{Para}$ is a mathematical stage for deep learning. It treats neural network layers as functions that depend on learnable parameters. Its rules for composition directly model how we stack layers to build a deep network, while 2-morphisms capture parameter sharing and reparameterization.

\begin{example}[Linear Layer as 1-morphism]
A linear neural network layer $x \mapsto Wx + b$ corresponds to the 1-morphism $(\mathbb{R}^{m \times n} \times \mathbb{R}^m, f)$ where $f((W,b), x) = Wx + b$.
\end{example}

Intuitively, composition in $\textbf{Para}$ corresponds to sequentially connecting network layers, and 2-morphisms formalize parameter sharing. For instance, weight sharing in CNNs corresponds to a 2-morphism from a full parameter space (independent weights per location) to a reduced one (a shared kernel), preserving the computational behavior. This structure cleanly separates learnable parameters ($P$) from computational logic ($f$).

To connect logical theories with neural architectures, we employ Lawvere theories, which provide the natural categorical semantics for algebraic structures and their operations \cite{hyland2007category}.

\begin{definition}[Lawvere Theory]
A \emph{Lawvere theory} $\mathbb{L}$ is a category with finite products where objects are natural numbers $\{0, 1, 2, \ldots\}$ representing arities, and morphisms encode the operations and algebraic laws of a theory as commutative diagrams.
\end{definition}

\textbf{Intuition:} A Lawvere theory is a ``blueprint'' for an algebraic system. It captures not just the operations (like AND and OR) but also the fundamental laws they must obey (like commutativity or De Morgan's laws) as a system of equations.

The following proposition establishes the bridge between logical theories and the parametric setting of neural networks:

\begin{proposition}[Lawvere Theories Embed in $\textbf{Para}$]
\label{prop:lawvere-para-embedding}
Every Lawvere theory $\mathbb{L}$ admits a canonical embedding into $\textbf{Para}$ via a functor $\mathcal{E}: \mathbb{L} \to \textbf{Para}$.
\end{proposition}

\begin{proof}
The proof, provided in Appendix A, proceeds by constructing the embedding for generating and structural morphisms and verifying functoriality.
\end{proof}

To make this correspondence concrete, we need to specify how logical theories are represented categorically:

\begin{definition}[Finitary Logical Theory]
A \emph{finitary logical theory} $\mathbb{T}$ consists of a finite signature $\Sigma$ of logical connectives (each with finite arity) and a finite set $\Phi$ of axioms relating these connectives.
\end{definition}

This definition encompasses a broad class of logical systems including propositional logic, many-valued logics, and bounded fragments of first-order logic.

\begin{definition}[Categorical Interpretation of Logical Expressions]
\label{def:categorical-interpretation}
Given a finitary logical theory $\mathbb{T}$ and its Lawvere theory $\mathbb{L}_{\mathbb{T}}$, an expression $E = c(E_1, \ldots, E_k)$ with free variables in $\{x_1, \ldots, x_n\}$ is interpreted as a morphism in $\mathbb{L}_{\mathbb{T}}$:
$$\llbracket c(E_1, \ldots, E_k) \rrbracket = c_{\mathbb{L}_{\mathbb{T}}} \circ \langle \llbracket E_1 \rrbracket, \ldots, \llbracket E_k \rrbracket \rangle: n \to 1$$
\end{definition}

\subsection{The Constructive Correspondence}

The heart of our framework is a systematic construction that transforms any finitary logical theory into a neural architecture with provable logical guarantees. We begin by establishing the consistency of this transformation:

\begin{lemma}[Consistency of Logical Axioms in Lawvere Theories]
\label{lem:consistency}
For a finitary logical theory $\mathbb{T}$ with semantic completeness, the set of categorical equations derived from its axioms is consistent if and only if $\mathbb{T}$ is logically consistent.
\end{lemma}

\begin{proof}
The proof (see Appendix B) establishes a bijective correspondence between the logical models of $\mathbb{T}$ and the categorical models of its Lawvere theory $\mathbb{L}_{\mathbb{T}}$.
\end{proof}

This consistency result enables our main correspondence theorem:

\begin{theorem}[Main Correspondence - Constructive Universality]
\label{thm:main-correspondence}
For every finitary logical theory $\mathbb{T}$ with semantic completeness, there exists a canonical, functorial construction yielding a neural architecture $\mathcal{N}_{\mathbb{T}}$ such that every forward pass corresponds to a valid derivation in $\mathbb{T}$.
\end{theorem}

\begin{proof}[Constructive Proof with Explicit Algorithms]
We establish the correspondence through three systematic algorithms (details in Appendix C). The construction proceeds by: (1) converting the logical theory $\mathbb{T}=(\Sigma, \Phi)$ into a Lawvere theory $\mathbb{L}_{\mathbb{T}}$ via categorical quotient construction, (2) mapping $\mathbb{L}_{\mathbb{T}}$ to a parametric model in $\textbf{Para}$ where logical connectives become differentiable operations and axioms become parameter constraints, and (3) realizing this parametric model as a neural network where composition structure mirrors logical derivation steps. The functoriality of this construction ensures that logical validity is preserved at each step.
\begin{enumerate}
    \item \textbf{Logic to Lawvere Theory:} A logical theory $\mathbb{T}=(\Sigma, \Phi)$ is converted into a Lawvere theory $\mathbb{L}_{\mathbb{T}}$ by creating a free category on the signature $\Sigma$ and quotienting by the categorical equations imposed by the axioms $\Phi$.
    \item \textbf{Lawvere Theory to Parametric Model:} The Lawvere theory $\mathbb{L}_{\mathbb{T}}$ is mapped to a model in $\textbf{Para}$ by assigning a differentiable parametric function $(P_c, f_c)$ to each logical connective $c \in \Sigma$. The logical axioms become algebraic constraints on the parameters.
    \item \textbf{Parametric Model to Neural Architecture:} The parametric model is realized as a neural network where layers correspond to logical connectives, network topology mirrors categorical composition, and training must enforce the extracted parameter constraints.
\end{enumerate}
\end{proof}

The uniqueness of this construction is established by the following:

\begin{proposition}[Uniqueness of Canonical Architecture]
\label{prop:uniqueness}
The neural architecture derived from a logical theory $\mathbb{T}$ is unique up to 2-morphisms in $\textbf{Para}$ (i.e., up to behavior-preserving reparameterization).
\end{proposition}

\begin{corollary}[Logical Correctness Guarantee]
\label{cor:logical-correctness}
Any neural network $\mathcal{N}_{\mathbb{T}}$ derived via our construction is guaranteed to respect the rules of the logical theory $\mathbb{T}$ by its very architecture.
\end{corollary}

\begin{proof}
The proof proceeds by strong induction on the depth of logical expressions, showing that the network's compositional structure, derived from the categorical model, mirrors the valid deductive steps of the logical theory. A forward pass through the network is thus computationally equivalent to a valid derivation in $\mathbb{T}$.
\end{proof}

This corollary establishes that our construction yields neural networks with mathematical guarantees about their reasoning behavior. Unlike traditional approaches where logical consistency must be learned from data, our networks are logically consistent by construction.

\section{Concrete Construction: Propositional Logic Networks}

We now demonstrate our theoretical framework by constructing neural architectures for propositional logic. This concrete instantiation illustrates how the abstract correspondence of Theorem~\ref{thm:main-correspondence} yields explicit neural architectures with mathematical guarantees about logical reasoning. Propositional logic is an ideal testbed: it is fundamental, computationally tractable, and rich enough to demonstrate our key principles.

\subsection{Propositional Logic as a Lawvere Theory}

We begin by encoding propositional logic with connectives $\{\land, \lor, \neg\}$ as a Lawvere theory, $\mathbb{L}_{\text{Bool}}$. The axioms of boolean algebra become commutative diagrams in $\mathbb{L}_{\text{Bool}}$, encoding laws like commutativity ($\land \circ \text{swap}_{2} = \land$) and De Morgan's laws ($\neg \circ \land = \lor \circ (\neg \times \neg)$) as categorical equations \cite{awodey2010category}. For instance, De Morgan's law $\neg(x \land y) = (\neg x) \lor (\neg y)$ corresponds to a commutative diagram where the composition $\neg \circ \land$ equals the composition $\lor \circ (\neg \times \neg)$, ensuring that both paths from the input $(x,y)$ to the final output yield identical results. This ensures our categorical encoding is faithful to classical propositional logic.

\subsection{Differentiable Boolean Neural Architecture}

To reconcile discrete boolean logic with continuous optimization, we design differentiable functions $\circ_\beta: [0,1]^n \to [0,1]$ that perfectly match boolean truth tables as a temperature parameter $\beta \to \infty$, while remaining differentiable for finite $\beta$.

\begin{proposition}[Continuous Relaxation of Boolean Operations]
\label{prop:continuous-relaxation}
The discrete boolean operations admit differentiable approximations derived from maximum-margin separation in logit space:
\begin{align}
\land_{\beta}(x, y) &= \sigma(\beta(x + y - 1.5)) \\
\lor_{\beta}(x, y) &= \sigma(\beta(x + y - 0.5)) \\
\neg_{\beta}(x) &= \sigma(\beta(1 - 2x))
\end{align}
where $\sigma(z) = (1 + e^{-z})^{-1}$ is the sigmoid function. The threshold constants are chosen to maximize separation between true and false cases: for conjunction, inputs $(1,1)$ yield $x+y-1.5 = 0.5 > 0$ (true), while all other boolean combinations yield $x+y-1.5 \leq -0.5 < 0$ (false). Similarly, for disjunction, only $(0,0)$ yields $x+y-0.5 = -0.5 < 0$ (false), while all other cases yield $x+y-0.5 \geq 0.5 > 0$ (true). For negation, the linear transformation $1-2x$ maps $0 \mapsto 1$ and $1 \mapsto -1$, achieving maximum separation.
\end{proposition}

\begin{proof}
Convergence as $\beta \to \infty$ is verified by examining the arguments of $\sigma$ for inputs in $\{0,1\}^n$. For instance, with $\land_{\beta}(1, 1)$, the argument is $0.5\beta \to \infty$, so $\sigma(0.5\beta) \to 1$. For $\land_{\beta}(1, 0)$, the argument is $-0.5\beta \to -\infty$, so $\sigma(-0.5\beta) \to 0$. Similar analysis holds for all cases.
\end{proof}

\begin{lemma}[Finite-$\beta$ Error Bounds]
\label{lem:error-bounds}
For $x, y \in \{0, 1\}$ and $\beta > 0$, the approximation errors are bounded exponentially:
\begin{align}
|\land_{\beta}(x, y) - \land(x, y)| \leq \frac{1}{1 + e^{0.5\beta}}, \quad |\neg_{\beta}(x) - \neg(x)| \leq \frac{1}{1 + e^{\beta}}
\end{align}
\end{lemma}

\begin{proof}
The proof, which involves a case-wise analysis for each boolean input, is provided in Appendix D.
\end{proof}

\begin{theorem}[Differentiable Logical Correctness]
\label{thm:differentiable-correctness}
A network $\mathcal{N}_{\beta}$ constructed from these layers is infinitely differentiable. Its Lipschitz constant is bounded by $(\frac{\beta\sqrt{2}}{4})^d$ for depth $d$, and it converges to the discrete logical function as $\beta \to \infty$.
\end{theorem}

\begin{proof}
The proof follows by a straightforward induction on network depth $d$, using the chain rule and the gradient bound $\|\nabla \circ_\beta\| \leq \frac{\beta\sqrt{2}}{4}$ for each layer.
\end{proof}

\textbf{Numerical Considerations:} While large $\beta$ ensures logical accuracy, it may cause numerical instability due to steep gradients. In practice, $\beta$ can be gradually increased during training (annealing) to balance optimization stability with logical precision.

\subsection{Weight Constraints and Optimization on Logical Manifolds}

The categorical axioms of $\mathbb{L}_{\text{Bool}}$ translate into algebraic constraints on the network's parameters, defining a manifold of logically-valid models. For example, commutativity of conjunction requires weight symmetry ($w_1 = w_2$), while De Morgan's laws impose complex nonlinear equations relating the parameters of $\land, \lor, \text{and } \neg$.

Training these networks requires optimizing a loss function $\mathcal{L}(W)$ subject to these constraints, $G(W)=0$. While one could use penalty methods to approximate the constraints by adding violation terms to the loss function, they offer no formal guarantees about constraint satisfaction during training. To ensure logical validity is strictly maintained throughout training, we must treat optimization as a constrained problem on the manifold $\mathcal{M}_{\mathbb{L}}$, for which Riemannian gradient descent is the principled approach.

\begin{definition}[Logical Constraint Manifold]
The \emph{logical constraint manifold} is the set $\mathcal{M}_{\mathbb{L}} = \{W \in \mathbb{R}^n : G_i(W) = 0, \forall i\}$, where $\{G_i\}$ are the constraint functions derived from logical axioms.
\end{definition}

To optimize on this manifold, we use Riemannian gradient descent \cite{absil2009optimization}. We endow $\mathcal{M}_{\mathbb{L}}$ with the standard Riemannian structure inherited from the ambient space $\mathbb{R}^n$. At each point $W \in \mathcal{M}_{\mathbb{L}}$, this defines a tangent space $T_W\mathcal{M}_{\mathbb{L}}$, a projection operator $P_W: \mathbb{R}^n \to T_W\mathcal{M}_{\mathbb{L}}$, and a retraction operator $\text{Retr}_W$ that maps points from the tangent space back to the manifold.

\begin{algorithm}{Riemannian Gradient Descent on Logical Manifolds}
\textbf{Input:} Loss $\mathcal{L}$, initial point $W_0 \in \mathcal{M}_{\mathbb{L}}$, learning rate $\eta$
\textbf{Output:} Optimized parameters $W^* \in \mathcal{M}_{\mathbb{L}}$

\textbf{Algorithm:} For $t = 1, 2, \ldots$:
\begin{enumerate}
    \item Compute Euclidean gradient: $\nabla \mathcal{L}(W_t)$
    \item Project to tangent space: $\xi_{W_t} = P_{W_t}(\nabla \mathcal{L}(W_t))$ where $P_{W_t}(v) = v - \sum_i \frac{\langle v, \nabla G_i(W_t) \rangle}{\|\nabla G_i(W_t)\|^2} \nabla G_i(W_t)$
    \item Update and retract: $W_{t+1} = \text{Retr}_{W_t}(W_t - \eta \xi_{W_t})$
    \item Verify constraint satisfaction: $\|G(W_{t+1})\| \leq \epsilon$ for tolerance $\epsilon$
\end{enumerate}
\textbf{Return} $W$
\end{algorithm}

\begin{theorem}[Convergence of Riemannian Optimization]
\label{thm:riemannian-convergence}
Under standard smoothness conditions on the loss function and manifold geometry, the Riemannian gradient descent algorithm is guaranteed to converge to a critical point of $\mathcal{L}$ on $\mathcal{M}_{\mathbb{L}}$.
\end{theorem}

This principled geometric approach ensures that network parameters remain on the manifold of logically-valid models throughout training. Alternative approaches, such as penalty-based methods that add constraint violations to the loss function, offer a trade-off between computational cost and exact constraint satisfaction but lack the guarantees of the Riemannian framework.

\subsection{Distributivity Law as Architectural Example}

We can construct an architecture enforcing the distributivity axiom $A \land (B \lor C) = (A \land B) \lor (A \land C)$, instantiating Theorem~\ref{thm:main-correspondence}. The architecture consists of two parallel computational arms:

\textbf{Left Arm:} Computes $A \land (B \lor C)$ via: (1) $\lor_\beta$ layer taking inputs $B, C$, (2) $\land_\beta$ layer combining output with input $A$.

\textbf{Right Arm:} Computes $(A \land B) \lor (A \land C)$ via: (1) two parallel $\land_\beta$ layers computing $A \land B$ and $A \land C$, (2) $\lor_\beta$ layer combining their outputs.

The categorical construction enforces parameter constraints through weight sharing: all $\land_\beta$ layers share identical parameters $(W_{\land}, b_{\land})$, and all $\lor_\beta$ layers share $(W_{\lor}, b_{\lor})$. Additionally, the constraint $\mathcal{L}_{\text{equiv}}(W) = \|\text{LeftArm}(A,B,C) - \text{RightArm}(A,B,C)\|^2 = 0$ must hold for all inputs, defining part of the logical constraint manifold. A case analysis for all boolean inputs $(A,B,C) \in \{0,1\}^3$ confirms the two arms produce identical outputs when the distributivity constraint is satisfied. This demonstrates how our framework translates logical laws into architectural constraints that guarantee correctness by design.

\section{Universality and Expressivity Results}

Having established the constructive correspondence, we now develop its fundamental mathematical properties. These results characterize the scope and limitations of our approach, providing precise bounds on what can be achieved through categorical construction.

\subsection{Universal Property of the Construction}

The universality of our construction is its most fundamental property, establishing that our categorical approach captures \emph{all} possible ways to build neural networks that respect a given logical theory.

\begin{definition}[Structurally Logical Neural Network]
\label{def:structurally-logical}
A neural network $\mathcal{N}$ is \emph{structurally logical} with respect to a theory $\mathbb{T} = (\Sigma, \Phi)$ if its architecture is a Directed Acyclic Graph of layers implementing operations from $\Sigma$, and its parameters satisfy the algebraic constraints derived from the axioms $\Phi$.
\end{definition}

For example, a network implementing boolean logic must have layers corresponding to $\{\land, \lor, \neg\}$ with parameters satisfying commutativity and De Morgan's law constraints.

\begin{theorem}[Structural Equivalence of Logically-Consistent Architectures]
\label{thm:structural-universality}
Every structurally logical neural network $\mathcal{N}$ arises uniquely as a model of the corresponding Lawvere theory $\mathbb{L}_{\mathbb{T}}$ in $\textbf{Para}$, up to natural isomorphism (reparameterization). This establishes our categorical construction as the canonical generator for the class of structurally logical architectures.
\end{theorem}

\begin{proof}[Proof Sketch]
We establish a bijection between structurally logical networks and models of $\mathbb{L}_{\mathbb{T}}$ in `Para`. The forward map extracts parametric morphisms from a network's layers, while the reverse map constructs a network from a model. The structural properties of the network (Definition \ref{def:structurally-logical}) ensure this mapping respects the functoriality and categorical equations of the model. The full proof, including verification that these maps are inverses, is in Appendix E.
\end{proof}

\begin{remark}[Practical Implications of Universality]
This theorem proves that \emph{any} structurally logical network must live on the constraint manifold defined by the logical axioms, providing mathematical justification for the specialized Riemannian optimization techniques developed in Section 3.
\end{remark}

\begin{corollary}[Completeness of the Framework]
Our construction provides a complete characterization of structurally logical neural architectures: every such architecture arises from our correspondence, and every architecture from our correspondence is structurally logical.
\end{corollary}

\begin{theorem}[Approximation Error Bounds with Temperature Dependence]
\label{thm:approximation-validity}
Let $\mathcal{N}_\beta$ be a neural network constructed via our framework with temperature $\beta$. For any logical axiom $(E_1 \equiv E_2)$ involving expressions of maximum depth $d$ and boolean inputs $x \in \{0,1\}^n$:
\begin{equation}
|\mathcal{N}_\beta^{E_1}(x) - \mathcal{N}_\beta^{E_2}(x)| \leq C \cdot \left(\frac{\beta\sqrt{2}}{4}\right)^d \cdot e^{-\alpha\beta}
\end{equation}
for constants $C, \alpha > 0$. For propositional logic, $\alpha \geq 0.5$. This bound combines two error sources: the Lipschitz factor $(\frac{\beta\sqrt{2}}{4})^d$ capturing error propagation through network depth, and the exponential term $e^{-\alpha\beta}$ reflecting the approximation quality of individual boolean operations (Lemma~\ref{lem:error-bounds}). The multiplicative combination arises because logical expressions are composed functions, where local approximation errors compound through composition according to the network's Lipschitz constant.
\end{theorem}

\begin{remark}[Practical Implications of Temperature Control]
This theorem provides a practical knob, $\beta$, for practitioners. A low $\beta$ allows for smoother gradients and easier initial training, while a high $\beta$ enforces near-perfect logical rigor, crucial for final deployment in safety-critical applications.
\end{remark}

\subsection{Compositional Properties and Expressivity}

A key advantage of the categorical approach is its natural handling of composition, enabling the modular construction of complex reasoning architectures.

\begin{proposition}[Logical Modularity]
\label{prop:modularity}
Given architectures $\mathcal{N}_1$ and $\mathcal{N}_2$ for theories $\mathbb{T}_1$ and $\mathbb{T}_2$:
\begin{enumerate}
    \item Sequential composition $\mathcal{N}_1 \circ \mathcal{N}_2$ implements logical pipelines.
    \item Parallel composition $\mathcal{N}_1 \times \mathcal{N}_2$ implements the product theory $\mathbb{T}_1 \times \mathbb{T}_2$, allowing simultaneous, independent reasoning in both systems.
\end{enumerate}
These compositions preserve logical correctness. For example, one can compose propositional and modal logic networks ($\mathcal{N}_{\text{Prop}} \times \mathcal{N}_{\text{Modal}}$) to reason about epistemic statements involving both truth values and knowledge operators.
\end{proposition}

The guarantees of our framework come at the cost of principled restrictions on expressivity.

\begin{theorem}[Expressivity via Term Algebras]
\label{thm:expressivity}
A neural architecture $\mathcal{N}_{\mathbb{T}}$ derived from theory $\mathbb{T}$ can express exactly the functions belonging to the term algebra $T_{\mathbb{T}}$ generated by $\mathbb{T}$.
\end{theorem}

\begin{proof}[Proof Sketch]
We construct an algebra isomorphism $\Psi: T_{\mathbb{T}}(\{0,1\}) \to \text{NetworkFunctions}(\mathcal{N}_{\mathbb{T}})$. The structural constraints imposed by our construction ensure that $\Psi$ is well-defined (maps equivalent logical expressions to the same network function), injective, and surjective.
\end{proof}

\begin{remark}[Expressivity Limitations and Benefits]
This theorem formalizes a fundamental trade-off: our networks gain logical verifiability by sacrificing the ability to learn functions outside their embedded logical theory. Unlike universal approximators, they cannot learn to violate their structural constraints. For example, a network built for standard boolean logic is architecturally incapable of learning a function that violates commutativity, regardless of the training data. This result clarifies the role of our architectures: they are not universal function approximators but ``specialized reasoners''. This is a feature for applications where verifiable reasoning is paramount, as it prevents the network from learning spurious, logically-invalid correlations from the data, ensuring its reasoning remains transparent and verifiable.
\end{remark}

\begin{corollary}[Separation of Logical Theories]
Networks derived from distinct logical theories $\mathbb{T}_1 \not\cong \mathbb{T}_2$ express different classes of functions, providing a principled basis for choosing logical foundations based on required expressivity.
\end{corollary}

\section{Extension of Categorical Deep Learning}

Our framework represents a fundamental extension of Categorical Deep Learning (CDL) beyond its current scope. While CDL has successfully characterized geometric constraints in neural architectures \cite{gavranovic2024categorical}, logical reasoning requires moving beyond the geometric symmetries that have been CDL's primary focus.

\begin{table}[h]
\centering
\caption{Geometric vs. Logical Categorical Deep Learning.}
\begin{tabular}{|l|l|l|}
\hline
\textbf{Feature} & \textbf{Geometric CDL (Existing)} & \textbf{Logical CDL (This Work)} \\
\hline
\textbf{Inductive Bias} & Geometric Symmetries & Semantic Rules \\
& (e.g., rotation, translation) & (e.g., logical axioms) \\
\hline
\textbf{Algebraic Structure} & Groups & Lawvere Theories \\
\hline
\textbf{Key Operations} & Group Actions (Invertible) & Logical Connectives (Irreversible) \\
\hline
\textbf{Categorical Model} & $G$-Sets / $G$-Representations & Models of Lawvere Theory \\
\hline
\textbf{Monad Type} & Group Action Monads & Free Algebra Monads \\
\hline
\textbf{Example Architecture} & Equivariant CNNs, GNNs & Differentiable Boolean Networks \\
\hline
\textbf{Verification} & Equivariance Testing & Logical Correctness Proofs \\
\hline
\textbf{Nature of Guarantee} & Equivariance Preservation & Logical Consistency \\
\hline
\textbf{Main Challenges} & High-dim. groups, representation & Constraint optimization, \\
& theory & discrete-continuous gap \\
\hline
\end{tabular}
\label{tab:cdl-comparison}
\end{table}

\subsection{Beyond Geometric Symmetries}

Traditional CDL leverages the rich theory of group actions \cite{bronstein2021geometric, cohen2016group}, but this encounters fundamental limitations when applied to logic. Logical inference often involves irreversible operations (e.g., modus ponens), which cannot be captured by group actions, as they are inherently invertible. Our framework handles this irreversibility through general categorical morphisms. Furthermore, CDL is typically descriptive, characterizing architectures that satisfy certain properties, whereas our approach is prescriptive, providing algorithms to derive architectures from logical specifications.

This extension is mathematically natural: where CDL uses group representations, we use models of Lawvere theories. Both are instances of using algebraic structure to constrain neural architectures, but logical structure requires the more general framework of universal algebra \cite{burris2012course}.

\subsection{Categorical Algebra Generalization}

Our approach maintains deep connections to the algebraic foundations of CDL. The key insight is that both geometric and logical constraints can be unified under the framework of monad algebras.

\begin{proposition}[Logical Theories as Generalized Monads]
\label{prop:logical-monads}
Every Lawvere theory $\mathbb{L}_{\mathbb{T}}$ generates a \emph{free algebra monad} $M_{\mathbb{T}}$ on the category of sets, and models of the theory correspond bijectively to algebras for this monad.
\end{proposition}

\begin{proof}
The proof, which involves verifying the monad laws for the construction $M_{\mathbb{T}}(X) = \coprod_{n \geq 0} \mathbb{L}_{\mathbb{T}}(n, 1) \times X^n$, is provided in Appendix F.
\end{proof}

\textbf{Intuition:} Monads can be understood as "formula-builders" that systematically construct complex expressions from simpler components. For propositional logic with variables $X = \{p, q\}$, the monad $M_{\mathbb{T}}(X)$ acts as a formula-building machine that generates all possible logical formulas from $p, q,$ and the logical connectives. The monad's unit $\eta$ maps a variable like $p$ to the atomic formula `p', and its multiplication $\mu$ performs substitution of formulas into other formulas, enabling compositional construction of complex logical expressions.

This proposition reveals that logical neural networks are instances of the general CDL pattern: they arise as algebras for monads, just like geometric neural networks. The difference lies in the nature of the monad—geometric monads arise from group actions, logical monads from algebraic theories.

\begin{corollary}[Unification with CDL and Broader Inductive Biases]
Geometric and logical neural networks are both instances of realizing algebras for appropriately chosen monads in $\textbf{Para}$. This unifies the design of networks with strong inductive biases. The choice of monad specifies the bias: a group action monad for geometry, a Lawvere theory monad for logic. This suggests systematic approaches for encoding other biases, such as those from temporal logic, probabilistic logic, or type theory.
\end{corollary}

\subsection{Hybrid Geometric-Logical Architectures}

This unification opens promising directions for architectures that respect both geometric and logical constraints. A \emph{distributive law} $\delta: G \circ M_{\mathbb{T}} \Rightarrow M_{\mathbb{T}} \circ G$ between a geometric monad $G$ and a logical monad $M_{\mathbb{T}}$ allows the construction of a composite monad that respects both constraint types.

This unification paves the way for multi-modal specifications where one could formally request an architecture that is simultaneously rotation-equivariant for processing 3D molecular data (a geometric constraint) and respects the axioms of a chemical logic for bond formation (a semantic constraint), deriving a single, principled architecture that embodies both physical and chemical principles.

\begin{example}[Equivariant Logical Graph Neural Network]
A GNN can be made both permutation-equivariant (a geometric constraint) and logically-constrained in its aggregation function (e.g., using $\land$ for intersection queries). The distributive law ensures that performing a logical aggregation and then permuting the graph is the same as permuting first and then aggregating. This enables principled reasoning on knowledge graphs where both relational structure and logical rules are essential.
\end{example}

This approach provides a foundation for architectures that are simultaneously geometrically equivariant and logically consistent, potentially enabling more robust reasoning in spatial and relational domains. Optimization for such hybrid architectures would involve navigating the intersection of geometric and logical constraint manifolds, a task for which the categorical foundations established here provide a systematic starting point.

\section{Scope, Limitations, and Future Work}

While our framework provides a principled foundation for logical neural networks, it is important to understand its scope, inherent limitations, and the research directions it opens. These reflect fundamental trade-offs between expressivity, tractability, and verifiability.

\subsection{Scope and Principled Trade-offs}

\textbf{Finitary Restriction:} Our results apply to \emph{finitary logical theories}—those with finite-arity connectives and a finite set of axioms. This restriction is not a limitation but a principled design choice that enables computational tractability while ensuring the corresponding Lawvere theory is well-behaved. As summarized in Table \ref{tab:logical-systems}, this scope includes a vast range of systems relevant to AI, such as propositional, modal, description, and many-valued logics, while necessarily excluding systems with infinitary axioms or quantification over infinite domains, like full first-order or higher-order logic. This restriction reflects the fundamental undecidability of such powerful logics \cite{church1936note}, making our finitary focus essential for achieving verifiable reasoning systems.

\begin{table}[h]
\centering
\caption{Logical Systems: Included vs. Excluded by Finitary Restriction.}
\begin{tabular}{|l|p{5.5cm}|p{5.5cm}|}
\hline
\textbf{Category} & \textbf{Included Systems} & \textbf{Excluded Systems} \\
\hline
\textbf{Propositional} & Classical and multi-valued logic & N/A \\
\hline
\textbf{Modal} & K, T, S4, S5, finite temporal logics & Infinitary modal logics \\
\hline
\textbf{First-Order} & Bounded quantification fragments & Full first-order logic, HOL \\
\hline
\textbf{Description} & $\mathcal{ALC}$, OWL fragments & Logics with infinite TBoxes \\
\hline
\end{tabular}
\label{tab:logical-systems}
\end{table}

\textbf{Expressivity Trade-off:} The mathematical guarantees of our framework come at the cost of restricted expressivity. As shown in Theorem \ref{thm:expressivity}, a network $\mathcal{N}_{\mathbb{T}}$ can only express functions definable within its underlying logical theory $\mathbb{T}$. Rather than viewing this as a limitation, this represents a deliberate trade-off for achieving verifiability: it makes behavior predictable and prevents the learning of logically invalid heuristics. This is analogous to the safety guarantees provided by strong typing in programming languages \cite{pierce2002types}.

\begin{remark}[Computational Complexity: A Crucial Distinction]
\textbf{Our networks perform \emph{expression evaluation}, which is efficient (linear in expression size), not \emph{satisfiability checking} (e.g., SAT), which is often computationally hard.} The framework remains tractable because it evaluates given logical expressions according to their truth conditions, rather than solving the general decision problem for the theory. For example, evaluating whether $(p \land q) \lor r$ is true given specific truth values for $p, q, r$ is straightforward, while determining whether a complex formula is satisfiable is NP-complete. This distinction is fundamental to understanding why our approach achieves both logical guarantees and computational tractability.
\end{remark}

\textbf{Positioning in Neuro-Symbolic Landscape:} Our "correctness-by-construction" approach offers a unique position in the neuro-symbolic landscape. Unlike semantic loss methods \cite{xu2018semantic, fischer2019dl2}, we provide hard guarantees, not soft encouragement. Unlike post-hoc hybrid systems \cite{manhaeve2018deepproblog}, we offer a unified, end-to-end differentiable model. And unlike specific implementations like Logic Tensor Networks, our framework is a universal, systematic construction for any finitary theory.

\subsection{Future Work: A Structured Research Program}

Our work establishes a foundation for a new class of verifiable learning systems and opens several avenues for future research, organized into three thematic areas:

\textbf{Practical Scalability:} The most immediate challenges involve scaling our methods to complex logical theories and practical applications:
\begin{enumerate}
    \item \textbf{Manifold Geometry:} Understanding the fine-grained geometry (curvature, topology) of logical constraint manifolds and their impact on optimization convergence.
    \item \textbf{Computational Algebraic Geometry:} Leveraging tools like Gröbner bases to develop more efficient projection and retraction operators for complex logical axioms.
    \item \textbf{Hybrid Manifold Optimization:} Developing algorithms to optimize effectively on intersections of geometric and logical constraint manifolds for hybrid architectures.
\end{enumerate}

\textbf{Semantic Grounding and Integration:} Bridging the gap between symbolic manipulation and perceptual understanding:
\begin{enumerate}
    \item \textbf{Symbol Grounding:} Interfacing logically-structured modules with perceptual front-ends (vision, language models) to learn meaningful logical representations from raw data.
    \item \textbf{Multi-Modal Integration:} Developing principled methods for combining logical reasoning with statistical learning in perception-reasoning pipelines.
\end{enumerate}

\textbf{Theoretical Extensions:} Expanding the scope and power of the categorical framework:
\begin{enumerate}
    \item \textbf{Probabilistic and Temporal Logic:} Extending the framework to handle uncertainty and temporal reasoning through appropriate categorical constructions.
    \item \textbf{Higher-Order Fragments:} Investigating bounded fragments of higher-order logic that remain finitary while increasing expressivity.
\end{enumerate}

Addressing these questions lies at the intersection of categorical logic, differential geometry, and optimization theory, and is essential for realizing the full potential of verifiable AI systems.

\section{Conclusion}

By providing a constructive and universal mapping from logic to architecture, our framework fundamentally transforms the design of reasoning systems. The central question shifts from an empirical search—"How can we train this network to be logical?"—to a formal specification problem: "What is the correct logic that this AI system should embody?"

Our contributions enable this transformation through four key advances. First, we provided an explicit, algorithmic bridge from any finitary logical theory to a neural network architecture, making logical specification directly actionable for system design. Second, our universality result (Theorem \ref{thm:structural-universality}) proves that this construction is canonical, capturing all possible ways to build structurally logical networks. Third, we extended Categorical Deep Learning beyond geometric symmetries to semantic rules, unifying both under the algebraic framework of monad algebras in the category of parametric maps. Fourth, we moved beyond pure theory by introducing a concrete methodology for training these highly-constrained networks through Riemannian optimization on logical constraint manifolds, transforming the abstract construction into a practical, trainable methodology.

Our framework makes a principled trade-off: it sacrifices the universal approximation power of standard networks for the far stronger guarantee of logical correctness by construction. The resulting networks are not merely trained to be rational, but are provably rational by design within their specified logical domains. While significant challenges remain, particularly in scaling the optimization and integrating with symbol grounding systems, our work provides both the theoretical charter and the practical foundations for building the next generation of AI systems: not just powerful pattern recognizers, but engines of verifiable reason.

\newpage
\appendix

\section*{Appendix: Technical Proofs and Algorithmic Details}

This appendix provides the detailed mathematical proofs and algorithmic specifications that underpin the results presented in the main text. Each section is designed to be as self-contained as possible, beginning with a restatement of the proposition or theorem being proved, followed by a high-level proof strategy and then the formal details. We assume familiarity with basic concepts from category theory.

\section{Proof of Proposition 2.1 (Lawvere Theory Embedding)}
\label{app:prop-lawvere-para-embedding}

\begin{proposition}
Every Lawvere theory $\mathbb{L}$ admits a canonical functor $\mathcal{E}: \mathbb{L} \to \textbf{Para}$ that preserves finite products.
\end{proposition}

\noindent\textbf{Proof Strategy.} We construct the functor $\mathcal{E}$ by defining its action on objects and morphisms of $\mathbb{L}$. The action on objects maps arities to vector spaces. The action on morphisms maps the generating operations of the theory to parameterized functions and structural morphisms (projections, diagonals) to fixed, non-parametric functions. We then verify that this construction preserves identity and composition, thus satisfying the definition of a functor.

\begin{proof}
We define the functor $\mathcal{E}: \mathbb{L} \to \textbf{Para}$ as follows:

\begin{enumerate}
    \item \textbf{Action on Objects:} For each object $n \in \text{Ob}(\mathbb{L})$, we define $\mathcal{E}(n) = \mathbb{R}^n$. This mapping respects the product structure of a Lawvere theory, as the coproduct of arities $m \oplus n = m+n$ in $\mathbb{L}$ maps to $\mathbb{R}^{m+n}$, which is canonically isomorphic to $\mathbb{R}^m \times \mathbb{R}^n = \mathcal{E}(m) \times \mathcal{E}(n)$ in \textbf{Set}, the category of objects for \textbf{Para}.

    \item \textbf{Action on Morphisms:} We define the action on morphisms by specifying its behavior on generating and structural morphisms.
    \begin{itemize}
        \item \textbf{Generating Operations:} For each generating operation $\omega: n \to 1$ in $\mathbb{L}$, which represents an $n$-ary operation of the algebraic theory, we define its image as a 1-morphism in \textbf{Para}:
        \[\mathcal{E}(\omega) := (P_\omega, f_\omega: P_\omega \times \mathbb{R}^n \to \mathbb{R})\]
        where $P_\omega$ is a designated parameter space (e.g., a subset of $\mathbb{R}^k$) and $f_\omega$ is a differentiable function representing the continuous relaxation of the operation $\omega$.
        \item \textbf{Structural Morphisms:} The structural morphisms inherent to any Lawvere theory (projections, diagonals, permutations) are mapped to fixed, parameter-free 1-morphisms (i.e., with a singleton parameter space $\{\star\}$):
        \begin{itemize}
            \item Projections $\pi_i^n: n \to 1$ map to $\mathcal{E}(\pi_i^n) := (\{\star\}, \text{proj}_i: \{\star\} \times \mathbb{R}^n \to \mathbb{R})$, where $\text{proj}_i((\star), (x_1, \dots, x_n)) = x_i$.
            \item Diagonals $\Delta^n: 1 \to n$ map to $\mathcal{E}(\Delta^n) := (\{\star\}, \text{diag}: \{\star\} \times \mathbb{R}^1 \to \mathbb{R}^n)$, where $\text{diag}((\star), x) = (x, \dots, x)$.
        \end{itemize}
        \item \textbf{General Morphisms:} Any morphism $\alpha: n \to m$ in a Lawvere theory is uniquely determined by its $m$ components, $\alpha = \langle \alpha_1, \dots, \alpha_m \rangle$, where each $\alpha_i: n \to 1$ is a morphism. This is a direct consequence of the universal property of the product object $m$. We define $\mathcal{E}(\alpha)$ using the product in \textbf{Para}:
        \[\mathcal{E}(\alpha) := \langle \mathcal{E}(\alpha_1), \dots, \mathcal{E}(\alpha_m) \rangle\]
        If $\mathcal{E}(\alpha_i) = (P_i, f_i)$, their product is the 1-morphism $(\prod_i P_i, \langle f_1, \dots, f_m \rangle)$.
    \end{itemize}

    \item \textbf{Verification of Functoriality:} We must show that $\mathcal{E}$ preserves identity and composition.
    \begin{itemize}
        \item \textbf{Identity Preservation:} The identity morphism $\text{id}_n: n \to n$ in $\mathbb{L}$ is given by $\langle \pi_1^n, \dots, \pi_n^n \rangle$. Applying $\mathcal{E}$ yields:
        \[\mathcal{E}(\text{id}_n) = \langle \mathcal{E}(\pi_1^n), \dots, \mathcal{E}(\pi_n^n) \rangle = \langle (\{\star\}, \text{proj}_1), \dots, (\{\star\}, \text{proj}_n) \rangle\]
        The product of these is $(\{\star\}^n, \langle \text{proj}_1, \dots, \text{proj}_n \rangle)$, which simplifies to $(\{\star\}, \text{id}_{\mathbb{R}^n})$. This is precisely the identity 1-morphism on the object $\mathcal{E}(n) = \mathbb{R}^n$ in \textbf{Para}.
        \item \textbf{Composition Preservation:} Let $\alpha: n \to m$ and $\beta: m \to k$ be morphisms in $\mathbb{L}$. Let $\mathcal{E}(\alpha)=(P,f)$ and $\mathcal{E}(\beta)=(Q,g)$. The composition in \textbf{Para} is $\mathcal{E}(\beta) \circ \mathcal{E}(\alpha) = (P \times Q, h)$ where $h((p,q), x) = g(q, f(p,x))$. The composition $\beta \circ \alpha$ in $\mathbb{L}$ is a morphism from $n \to k$. By the definition of composition in a Lawvere theory, the $j$-th component of this composite, $(\beta \circ \alpha)_j: n \to 1$, is obtained by substituting the outputs of $\alpha$ into the $j$-th component of $\beta$. Applying our functor $\mathcal{E}$ to $\beta \circ \alpha$ yields a parametric map that computes $g(q, f(p,x))$ with parameters $(p,q)$, which is exactly the \textbf{Para}-composition. The correspondence holds due to the compositional definition of morphisms in $\mathbb{L}$ and our construction.
    \end{itemize}
\end{enumerate}
Thus, $\mathcal{E}$ is a product-preserving functor.
\end{proof}

\section{Proof of Lemma 2.3 (Consistency Equivalence)}
\label{app:lem-consistency}

\begin{lemma}
For a finitary logical theory $\mathbb{T}$ with semantic completeness, the set of categorical equations derived from its axioms is consistent if and only if $\mathbb{T}$ is logically consistent.
\end{lemma}

\noindent\textbf{Proof Strategy.} We prove the two directions of the biconditional by constructing explicit mappings between the models of the logical theory and the models of its corresponding Lawvere theory. The assumption of semantic completeness for $\mathbb{T}$ (i.e., syntax matches semantics, $\vdash \phi \iff \models \phi$) is the crucial link that guarantees this correspondence.

\begin{proof}
\textbf{Part 1: ($\mathbb{T}$ is logically consistent $\implies$ Categorical equations are consistent).}
\begin{enumerate}
    \item Assume $\mathbb{T}$ is logically consistent. By definition, this means there exists at least one model for $\mathbb{T}$. By the assumption of semantic completeness, this model can be taken to be a set-theoretic model.
    \item Let $\mathcal{M}$ be such a model. It consists of a domain set $D$ and an interpretation function that assigns a function $\mathcal{M}(c): D^k \to D$ to each $k$-ary connective $c \in \Sigma$, such that all axioms in $\Phi$ are satisfied.
    \item We construct a categorical model, which is a product-preserving functor $M: \mathbb{L}_{\mathbb{T}} \to \textbf{Set}$.
    \begin{itemize}
        \item \textbf{On Objects:} For an object $n \in \text{Ob}(\mathbb{L}_{\mathbb{T}})$, set $M(n) := D^n$. This preserves products since $M(m+n) = D^{m+n} \cong D^m \times D^n = M(m) \times M(n)$.
        \item \textbf{On Morphisms:} For a generating morphism $c_{\mathbb{L}_{\mathbb{T}}}: k \to 1$ in the Lawvere theory, set its image $M(c_{\mathbb{L}_{\mathbb{T}}})$ to be the function $\mathcal{M}(c): D^k \to D$.
    \end{itemize}
    \item Since $\mathcal{M}$ is a valid logical model, it satisfies all axioms in $\Phi$. For an axiom $E_1 \equiv E_2$, this means the functions interpreting $E_1$ and $E_2$ are identical. This implies that the functor $M$ maps the corresponding morphisms $\llbracket E_1 \rrbracket$ and $\llbracket E_2 \rrbracket$ to the same function in \textbf{Set}, thereby respecting the categorical equation $\llbracket E_1 \rrbracket = \llbracket E_2 \rrbracket$.
    \item Therefore, $M$ is a valid product-preserving functor, meaning a categorical model for $\mathbb{L}_{\mathbb{T}}$ exists. The existence of a model implies that the defining equations are consistent.
\end{enumerate}

\textbf{Part 2: (Categorical equations are consistent $\implies \mathbb{T}$ is logically consistent).}
\begin{enumerate}
    \item Assume the categorical equations defining $\mathbb{L}_{\mathbb{T}}$ are consistent. By definition, this means there exists a model, i.e., a product-preserving functor $M: \mathbb{L}_{\mathbb{T}} \to \textbf{Set}$.
    \item We construct a logical model $\mathcal{M}$ for the theory $\mathbb{T}$.
    \begin{itemize}
        \item \textbf{Domain:} Let the domain of $\mathcal{M}$ be the set $D := M(1)$.
        \item \textbf{Interpretation:} For each $k$-ary connective $c \in \Sigma$, define its interpretation $\mathcal{M}(c)$ to be the function $M(c_{\mathbb{L}_{\mathbb{T}}}): M(k) \to M(1)$. Since $M$ preserves products, $M(k) \cong (M(1))^k = D^k$, so $\mathcal{M}(c)$ is a function from $D^k$ to $D$.
    \end{itemize}
    \item We must show that $\mathcal{M}$ satisfies the axioms in $\Phi$. Consider an axiom $E_1 \equiv E_2$. In the Lawvere theory $\mathbb{L}_{\mathbb{T}}$, this corresponds to the categorical equation $\llbracket E_1 \rrbracket = \llbracket E_2 \rrbracket$.
    \item Since $M$ is a functor, it preserves this equality: $M(\llbracket E_1 \rrbracket) = M(\llbracket E_2 \rrbracket)$. By our construction of $\mathcal{M}$, this equality of functions in \textbf{Set} means that the interpretation of $E_1$ under $\mathcal{M}$ is identical to the interpretation of $E_2$.
    \item This holds for all axioms in $\Phi$. Therefore, $\mathcal{M}$ is a valid logical model for $\mathbb{T}$. The existence of such a model proves that $\mathbb{T}$ is logically consistent.
\end{enumerate}
\end{proof}

\section{Constructive Algorithms (for Theorem 2.4)}
\label{app:thm-main-correspondence}

The following three algorithms form the constructive core of our framework. Together, they provide a deterministic procedure for transforming a high-level logical specification into a concrete, trainable neural network whose architecture enforces the specified logic.

\begin{algorithm}[h!]
\caption{Logic to Lawvere Theory Construction}
\begin{algorithmic}[1]
\STATE \textbf{Input:} Finitary logical theory $\mathbb{T} = (\Sigma, \Phi)$, where $\Sigma$ is a signature of connectives and $\Phi$ is a set of axioms.
\STATE \textbf{Output:} The corresponding Lawvere theory $\mathbb{L}_{\mathbb{T}}$.
\STATE
\STATE \textbf{Free Theory Construction:} Construct the free category with finite products, $\mathbb{F}(\Sigma)$, generated by the signature $\Sigma$. The objects are natural numbers, and morphisms are all possible well-formed expressions built from the connectives in $\Sigma$ and structural morphisms (projections, diagonals). \COMMENT{Create the 'raw material' of all possible expressions.}
\STATE \textbf{Axiom Translation:} For each axiom $\phi \in \Phi$ of the form $E_1 \equiv E_2$ (where $E_1, E_2$ are expressions with $n$ free variables), interpret both sides as morphisms $n \to 1$ in $\mathbb{F}(\Sigma)$ using Definition \ref{def:categorical-interpretation}. This yields a set of categorical equations $\{\llbracket E_1 \rrbracket = \llbracket E_2 \rrbracket\}_{\phi \in \Phi}$. \COMMENT{Turn logical laws into algebraic identities.}
\STATE \textbf{Quotient Construction:} Define $\mathbb{L}_{\mathbb{T}}$ as the quotient category $\mathbb{F}(\Sigma) / \sim$, where $\sim$ is the congruence relation generated by the set of categorical equations from the previous step. In this quotient, logically equivalent expressions become equal morphisms. \COMMENT{'Bake' the laws into the structure by identifying equivalent expressions.}
\STATE \textbf{Return} $\mathbb{L}_{\mathbb{T}}$
\end{algorithmic}
\end{algorithm}

\begin{algorithm}[h!]
\caption{Lawvere Theory to Parametric Model in \textbf{Para}}
\begin{algorithmic}[1]
\STATE \textbf{Input:} Lawvere theory $\mathbb{L}_{\mathbb{T}}$.
\STATE \textbf{Output:} A parametric model (functor) $M: \mathbb{L}_{\mathbb{T}} \to \textbf{Para}$.
\STATE
\STATE \textbf{Object Assignment:} Define the action of the model on objects by setting $M(n) = \mathbb{R}^n$ for each object $n \in \text{Ob}(\mathbb{L}_{\mathbb{T}})$. \COMMENT{Map abstract arities to concrete vector spaces.}
\STATE \textbf{Operation Implementation:} For each generating morphism $c_{\mathbb{L}_{\mathbb{T}}}: n \to 1$ in the theory, assign a parametric map $M(c_{\mathbb{L}_{\mathbb{T}}}) := (P_c, f_c)$, where $f_c: P_c \times \mathbb{R}^n \to \mathbb{R}$ is a differentiable function. \COMMENT{Realize each logical connective as a differentiable, parameterized function.}
\STATE \textbf{Constraint Extraction:} For each axiom $E_1 \equiv E_2$ that defined $\mathbb{L}_{\mathbb{T}}$, the equation $M(\llbracket E_1 \rrbracket) = M(\llbracket E_2 \rrbracket)$ must hold. This equality of parametric maps imposes an algebraic constraint equation $G(W) = 0$ on the collection $W$ of all parameters from the set $\{P_c\}$. \COMMENT{Translate algebraic identities into equations that the parameters must satisfy.}
\STATE \textbf{Return} The model $M$ (i.e., the collection of parametric maps and their associated constraints).
\end{algorithmic}
\end{algorithm}

\begin{algorithm}[h!]
\caption{Parametric Model to Neural Architecture}
\begin{algorithmic}[1]
\STATE \textbf{Input:} A parametric model $M \in \text{Mod}(\mathbb{L}_{\mathbb{T}}, \textbf{Para})$.
\STATE \textbf{Output:} A trainable neural network $\mathcal{N}_{\mathbb{T}}$.
\STATE
\STATE \textbf{Layer Construction:} For each distinct parametric morphism $(P_c, f_c) = M(c_{\mathbb{L}_{\mathbb{T}}})$ in the model, instantiate a corresponding neural network layer whose forward pass computes $f_c$ and whose trainable weights are the parameters in $P_c$. \COMMENT{Create a layer for each logical operation.}
\STATE \textbf{Network Topology:} To evaluate a complex logical expression $E$, compose the layers according to the structure of $E$. The categorical composition in $\mathbb{L}_{\mathbb{T}}$ translates directly into a directed acyclic graph of layers. \COMMENT{Wire the layers together to mirror the logical formula.}
\STATE \textbf{Constraint Enforcement:} The parameter constraints $G(W)=0$ extracted in Algorithm 6 define the logical constraint manifold $\mathcal{M}_{\mathbb{L}}$. Train the network's parameters $W$ using an optimization algorithm that respects this manifold, such as Riemannian gradient descent (Algorithm 4), ensuring that the network remains logically valid throughout training. \COMMENT{Train the network on the manifold of logically valid models.}
\STATE \textbf{Return} $\mathcal{N}_{\mathbb{T}}$
\end{algorithmic}
\end{algorithm}

\section{Proof of Lemma 3.2 (Finite-$\beta$ Error Bounds)}
\label{app:lem-error-bounds}

\begin{lemma}
For $x, y \in \{0, 1\}$ and $\beta > 0$, the approximation errors of the continuous relaxations are bounded exponentially:
\[|\land_{\beta}(x, y) - \land(x, y)| \leq \frac{1}{1 + e^{0.5\beta}}, \quad |\neg_{\beta}(x) - \neg(x)| \leq \frac{1}{1 + e^{\beta}}\]
\end{lemma}

\begin{proof}
We perform a case-wise analysis of the absolute error $|\circ_\beta(\mathbf{x}) - \circ(\mathbf{x})|$ for each operation over all combinations of boolean inputs $\mathbf{x} \in \{0, 1\}^n$. We use the identity $\sigma(-z) = 1 - \sigma(z)$.

\noindent\textbf{1. Analysis for Conjunction ($\land_\beta$):}
Let $f(x,y) = \land_\beta(x,y) = \sigma(\beta(x+y-1.5))$.
\begin{itemize}
    \item \textbf{Case (x,y) = (0,0):} True value is 0. Network output is $f(0,0) = \sigma(-1.5\beta)$. \\
    Error is $|\sigma(-1.5\beta) - 0| = \frac{1}{1 + e^{1.5\beta}}$.
    \item \textbf{Case (x,y) = (0,1):} True value is 0. Network output is $f(0,1) = \sigma(-0.5\beta)$. \\
    Error is $|\sigma(-0.5\beta) - 0| = \frac{1}{1 + e^{0.5\beta}}$.
    \item \textbf{Case (x,y) = (1,0):} True value is 0. Network output is $f(1,0) = \sigma(-0.5\beta)$. \\
    Error is $|\sigma(-0.5\beta) - 0| = \frac{1}{1 + e^{0.5\beta}}$.
    \item \textbf{Case (x,y) = (1,1):} True value is 1. Network output is $f(1,1) = \sigma(0.5\beta)$. \\
    Error is $|\sigma(0.5\beta) - 1| = |-(1-\sigma(0.5\beta))| = \sigma(-0.5\beta) = \frac{1}{1 + e^{0.5\beta}}$.
\end{itemize}
The term $\frac{1}{1 + e^{0.5\beta}}$ is strictly greater than $\frac{1}{1 + e^{1.5\beta}}$ for $\beta > 0$. Thus, the maximum error is bounded by $\frac{1}{1 + e^{0.5\beta}}$.

\noindent\textbf{2. Analysis for Negation ($\neg_\beta$):}
Let $g(x) = \neg_\beta(x) = \sigma(\beta(1-2x))$.
\begin{itemize}
    \item \textbf{Case x = 0:} True value is 1. Network output is $g(0) = \sigma(\beta)$. \\
    Error is $|\sigma(\beta) - 1| = \sigma(-\beta) = \frac{1}{1 + e^{\beta}}$.
    \item \textbf{Case x = 1:} True value is 0. Network output is $g(1) = \sigma(-\beta)$. \\
    Error is $|\sigma(-\beta) - 0| = \frac{1}{1 + e^{\beta}}$.
\end{itemize}
The maximum error is bounded by $\frac{1}{1 + e^{\beta}}$. The analysis for disjunction $\lor_\beta$ is analogous to conjunction and yields the same bound.
\end{proof}

\section{Proof of Theorem 4.1 (Structural Equivalence)}
\label{app:thm-structural-universality}

\begin{theorem}
Every structurally logical neural network $\mathcal{N}$ arises uniquely as a model of the corresponding Lawvere theory $\mathbb{L}_{\mathbb{T}}$ in $\textbf{Para}$, up to natural isomorphism (reparameterization). This establishes our categorical construction as the canonical generator for the class of structurally logical architectures.
\end{theorem}

\noindent\textbf{Prerequisite Definition.} We recall from the main text:
\begin{definition}[Structurally Logical Neural Network - SLNN]
A neural network $\mathcal{N}$ is \emph{structurally logical} with respect to a theory $\mathbb{T} = (\Sigma, \Phi)$ if its architecture is a Directed Acyclic Graph of layers implementing operations from $\Sigma$, and its parameters satisfy the algebraic constraints derived from the axioms $\Phi$.
\end{definition}

\noindent\textbf{Proof Strategy.} The proof establishes a categorical equivalence between the category of SLNNs (where morphisms are architecture-preserving maps) and the category of models of $\mathbb{L}_{\mathbb{T}}$ in \textbf{Para} (where morphisms are natural transformations). We achieve this by constructing functors in both directions, $F: \text{SLNN} \to \text{Mod}(\mathbb{L}_{\mathbb{T}}, \textbf{Para})$ and $G: \text{Mod}(\mathbb{L}_{\mathbb{T}}, \textbf{Para}) \to \text{SLNN}$, and showing they form an equivalence of categories.

\begin{proof}
\textbf{1. The Functor $F$: From Network to Model ($F: \text{SLNN} \to \text{Mod}(\mathbb{L}_{\mathbb{T}}, \textbf{Para})$)}
\begin{itemize}
    \item \textbf{Action on Objects (Networks):} Given a network $\mathcal{N} \in \text{SLNN}$, we construct a functor (model) $M_{\mathcal{N}}: \mathbb{L}_{\mathbb{T}} \to \textbf{Para}$.
    \begin{itemize}
        \item $M_{\mathcal{N}}$ maps an object $n$ in $\mathbb{L}_{\mathbb{T}}$ to the object $\mathbb{R}^n$ in \textbf{Para}.
        \item For each generating morphism $c:k \to 1$ in $\mathbb{L}_{\mathbb{T}}$, Definition 4.1 guarantees that $\mathcal{N}$ contains a corresponding layer. We define $M_{\mathcal{N}}(c)$ to be the parametric morphism $(P_c, f_c)$ represented by that layer.
    \end{itemize}
    \item \textbf{Verification of Model Axioms:} We must show that $M_{\mathcal{N}}$ is a valid model, i.e., a product-preserving functor that respects the relations of $\mathbb{L}_{\mathbb{T}}$.
    \begin{itemize}
        \item It is product-preserving by construction.
        \item By Definition 4.1, the parameters of $\mathcal{N}$ must satisfy the algebraic constraints from the axioms $\Phi$. These are precisely the constraints needed to ensure that for any axiom $E_1 \equiv E_2$, the composite parametric map for $E_1$ is identical to that for $E_2$. This directly implies that the functor $M_{\mathcal{N}}$ respects the defining equations of $\mathbb{L}_{\mathbb{T}}$. Thus, $M_{\mathcal{N}}$ is a valid model.
    \end{itemize}
    \item \textbf{Action on Morphisms (Network Maps):} A morphism between SLNNs is an architecture-preserving map that preserves the layer functions. $F$ maps this to a natural transformation between the corresponding models in \textbf{Para}.
\end{itemize}

\textbf{2. The Functor $G$: From Model to Network ($G: \text{Mod}(\mathbb{L}_{\mathbb{T}}, \textbf{Para}) \to \text{SLNN}$)}
\begin{itemize}
    \item \textbf{Action on Objects (Models):} Given a model $M: \mathbb{L}_{\mathbb{T}} \to \textbf{Para}$, we construct a network $\mathcal{N}_M$.
    \begin{itemize}
        \item For each generating morphism $c:k \to 1$ in $\mathbb{L}_{\mathbb{T}}$, we create a network layer that implements its image under the model, the parametric morphism $M(c)=(P_c, f_c)$.
        \item The network topology for any complex expression is determined by the composition of morphisms in $\mathbb{L}_{\mathbb{T}}$, resulting in a DAG of these layers.
    \end{itemize}
    \item \textbf{Verification of SLNN Property:} We must show that $\mathcal{N}_M$ is an SLNN.
    \begin{itemize}
        \item The network is a DAG of layers from $\Sigma$ by construction.
        \item The parameters of the network are precisely the parameters of the model $M$. Since $M$ is a model, it must respect the defining equations of $\mathbb{L}_{\mathbb{T}}$. This means its parameters automatically satisfy the constraints derived from the axioms $\Phi$.
        \item Therefore, $\mathcal{N}_M$ satisfies Definition 4.1 and is an SLNN.
    \end{itemize}
    \item \textbf{Action on Morphisms (Natural Transformations):} $G$ maps a natural transformation between models to a corresponding architecture-preserving map between the generated networks.
\end{itemize}

\textbf{3. Verification of Equivalence ($F \circ G \simeq \text{Id}$ and $G \circ F \simeq \text{Id}$)}
\begin{itemize}
    \item \textbf{$G(F(\mathcal{N})) \cong \mathcal{N}$:} Applying $F$ to a network $\mathcal{N}$ extracts its layer-wise parametric maps to form a model $M_\mathcal{N}$. Applying $G$ to this model reconstructs a network $\mathcal{N}'$ from these exact same maps with the same topology. The resulting network $\mathcal{N}'$ is behaviorally identical to $\mathcal{N}$ and is isomorphic to it in the category of SLNNs. Any difference would be a reparameterization, which is an isomorphism (a 2-morphism in \textbf{Para}).
    \item \textbf{$F(G(M)) \cong M$:} Applying $G$ to a model $M$ creates a network $\mathcal{N}_M$ whose layers are defined by the parametric maps $M(c)$. Applying $F$ to $\mathcal{N}_M$ extracts these very same maps to build a new model $M'$. Since both models $M$ and $M'$ are functors determined by their action on the generators of $\mathbb{L}_{\mathbb{T}}$, and their actions on these generators are identical by construction, the models are naturally isomorphic.
\end{itemize}
The functors $F$ and $G$ establish an equivalence of categories, proving that every SLNN arises from our construction (up to isomorphism) and that the construction is canonical.
\end{proof}

\section{Proof of Proposition 5.1 (Logical Theories as Monads)}
\label{app:prop-logical-monads}

\begin{proposition}
Every Lawvere theory $\mathbb{L}_{\mathbb{T}}$ generates a free algebra monad $M_{\mathbb{T}}$ on the category of sets, and models of the theory correspond bijectively to algebras for this monad.
\end{proposition}

\noindent\textbf{Intuitive Reading:} This proposition states that the rules of a logical theory can be packaged into a "formula-building machine" (a monad). A set $X$ of variables goes in, and the set $M_{\mathbb{T}}(X)$ of all possible well-formed formulas one can build comes out. The monad structure provides a way to substitute variables (`unit`) and substitute formulas into other formulas (`multiplication`). 'Algebras' for this monad are then sets where these formulas can be consistently evaluated, which is exactly what a logical model does.

\begin{proof}
Let $\mathbb{L}_{\mathbb{T}}$ be a Lawvere theory. We construct the monad $M_{\mathbb{T}} = (M_{\mathbb{T}}, \eta, \mu)$ on the category \textbf{Set}.

\begin{enumerate}
    \item \textbf{The Functor $M_{\mathbb{T}}: \textbf{Set} \to \textbf{Set}$:}
    \begin{itemize}
        \item \textbf{On Objects:} For any set $X$ (of "variables"), we define $M_{\mathbb{T}}(X)$ as the set of all terms of the theory with variables from $X$. Formally, this is the disjoint union:
        \[ M_{\mathbb{T}}(X) := \coprod_{n \geq 0} \text{Hom}_{\mathbb{L}_{\mathbb{T}}}(n, 1) \times X^n \]
        An element of this set is a pair $(\alpha, (x_1, \dots, x_n))$, where $\alpha: n \to 1$ is an $n$-ary operation from the theory and $(x_1, \dots, x_n)$ is a tuple of variables. This represents the term $\alpha(x_1, \dots, x_n)$.
        \item \textbf{On Morphisms:} For a function $f: X \to Y$, we define $M_{\mathbb{T}}(f): M_{\mathbb{T}}(X) \to M_{\mathbb{T}}(Y)$ as the map that substitutes variables.
        $M_{\mathbb{T}}(f)(\alpha, (x_1, \dots, x_n)) := (\alpha, (f(x_1), \dots, f(x_n)))$.
    \end{itemize}

    \item \textbf{The Unit Natural Transformation $\eta: \text{Id} \Rightarrow M_{\mathbb{T}}$:}
    The component at a set $X$, $\eta_X: X \to M_{\mathbb{T}}(X)$, maps a variable $x \in X$ to the simplest term representing just that variable.
    \[ \eta_X(x) := (\text{id}_1, (x)) \]
    where $\text{id}_1: 1 \to 1$ is the identity morphism in $\mathbb{L}_{\mathbb{T}}$.

    \item \textbf{The Multiplication Natural Transformation $\mu: M_{\mathbb{T}} \circ M_{\mathbb{T}} \Rightarrow M_{\mathbb{T}}$:}
    The component at a set $X$, $\mu_X: M_{\mathbb{T}}(M_{\mathbb{T}}(X)) \to M_{\mathbb{T}}(X)$, performs substitution of complex terms into other terms, effectively "flattening" a term-of-terms. An element of $M_{\mathbb{T}}(M_{\mathbb{T}}(X))$ is a term $\alpha(t_1, \dots, t_n)$, where each $t_i$ is itself a term $(\beta_i, \mathbf{x}_i)$. The map $\mu_X$ composes these morphisms in the Lawvere theory to produce a single, flattened term. This corresponds to the composition of morphisms in $\mathbb{L}_{\mathbb{T}}$.
\end{enumerate}

\noindent\textbf{Verification of Monad Laws:}
\begin{itemize}
    \item \textbf{Right Unit Law ($\mu_X \circ M_{\mathbb{T}}(\eta_X) = \text{id}_{M_{\mathbb{T}}(X)}$):}
    This law corresponds to substituting variables into a term, which should yield the original term. $M_{\mathbb{T}}(\eta_X)$ takes a term $\alpha(x_1, \dots, x_n)$ and maps it to a term-of-terms where each variable $x_i$ is replaced by the term $(\text{id}_1, (x_i))$. Applying $\mu_X$ then flattens this back, resulting in the original term.
    \emph{Intuition:} \texttt{substitute($\alpha$(x, y), [x -> "x", y -> "y"]) = $\alpha$(x, y)}.

    \item \textbf{Left Unit Law ($\mu_X \circ \eta_{M_{\mathbb{T}}(X)} = \text{id}_{M_{\mathbb{T}}(X)}$):}
    This law corresponds to taking a term $t$ and substituting it into the identity operation. $\eta_{M_{\mathbb{T}}(X)}$ maps a term $t$ to the term-of-terms $(\text{id}_1, (t))$. Applying $\mu_X$ flattens this, yielding $t$ back.
    \emph{Intuition:} \texttt{substitute("x", [x -> $\alpha$(y,z)]) = $\alpha$(y,z)}.

    \item \textbf{Associativity Law ($\mu_X \circ M_{\mathbb{T}}(\mu_X) = \mu_X \circ \mu_{M_{\mathbb{T}}(X)}$):}
    This corresponds to the associativity of substitution. Whether one substitutes terms-of-terms into a larger term first and then flattens everything, or flattens the inner terms first before substituting, the result is the same. This property is guaranteed by the associativity of morphism composition in the Lawvere theory $\mathbb{L}_{\mathbb{T}}$.
\end{itemize}

Finally, the bijective correspondence between models of a Lawvere theory $\mathbb{L}_{\mathbb{T}}$ in \textbf{Set} and algebras for its free monad $M_{\mathbb{T}}$ is a standard, fundamental result in categorical universal algebra \cite{hyland2007category}. An $M_{\mathbb{T}}$-algebra is a set $A$ with a map $M_{\mathbb{T}}(A) \to A$ that is compatible with the monad structure; this map is precisely what defines the action of each logical operator on the set $A$, making it a logical model.
\end{proof}

\end{document}